\documentclass[a4paper,UKenglish]{article}

\input{preamble.tex}

\title{The Parameterized Hardness of the k-Center Problem in Transportation 
Networks}

\author[1]{Andreas Emil Feldmann\footnote{Supported by the Czech Science 
Foundation GA\v{C}R (grant \#19-27871X), and by the Center for 
Foundations of Modern Computer Science (Charles Univ.\ project UNCE/SCI/004).}}
\author[2]{D\'aniel Marx\footnote{Supported by ERC Consolidator Grant SYSTEMATICGRAPH (No.~725978)}}
\affil[1]{Department of Applied Mathematics, Charles University, Prague, 
Czechia\\ \texttt{feldmann.a.e@gmail.com}}
\affil[2]{Institute for Computer Science and Control, Hungarian Academy of 
Sciences\\ (MTA SZTAKI) \texttt{dmarx@cs.bme.hu}}

\date{}

\newcommand{\kc}{\pname{$k$-Center}\xspace}
\newcommand{\gti}{\pname{GT$_\leq$}\xspace}
\DeclareMathOperator{\dist}{dist}

\newcommand{\polyn}{\cdot n^{O(1)}}

\begin{document}

\maketitle

\begin{abstract}
In this paper we study the hardness of the \kc problem on inputs that model 
transportation networks. For the problem, a graph $G=(V,E)$ with edge lengths 
and an integer $k$ are given and a center set $C\subseteq V$ needs to be chosen 
such that $|C|\leq k$. The aim is to minimize the maximum distance of any 
vertex 
in the graph to the closest center. This problem arises in many applications of 
logistics, and thus it is natural to consider inputs that model transportation 
networks. Such inputs are often assumed to be planar graphs, low doubling 
metrics, or bounded highway dimension graphs. For each of these models, 
parameterized approximation algorithms have been shown to exist. We complement 
these results by proving that the \kc problem is W[1]\hy{}hard on planar graphs 
of constant doubling dimension, where the parameter is the combination of the 
number of centers $k$, the highway dimension $h$, and the pathwidth $p$. 
Moreover, under the Exponential Time Hypothesis there is no $f(k,p,h)\cdot 
n^{o(p+\sqrt{k+h})}$ time algorithm for any computable function $f$. Thus it is 
unlikely that the optimum solution to \kc can be found efficiently, even when 
assuming that the input graph abides to all of the above models for 
transportation networks at once!

Additionally we give a simple parameterized $(1+\eps)$-approximation algorithm 
for inputs of doubling dimension $d$ with runtime $(k^k/\eps^{O(kd)})\polyn$. 
This generalizes a previous result, which considered inputs in $D$-dimensional 
$L_q$ metrics.
\end{abstract}

\section{Introduction}

Given a graph $G=(V,E)$ with positive edge lengths $\ell:E\rightarrow 
\mathbb{Q}^+$, the \kc problem asks to find $k$ center vertices such that every 
vertex of the graph is as close as possible to one of the centers. More 
formally, a solution to \kc is a set $C\subseteq V$ of \emph{centers} such that 
$|C|\leq k$. If $\dist(u,v)$ denotes the length of the shortest path between 
$u$ and $v$ according to the edge lengths $\ell$, the objective is to minimize 
the \emph{cost} $\rho=\max_{u\in V}\min_{v\in C}\dist(u,v)$ of the solution 
$C$. While this is the standard way of defining the problem, throughout this 
paper we will rather think of it as covering the graph with balls of minimum 
radius. That is, let $B_v(r)=\{u\in V\mid\dist(u,v)\leq r\}$ be the \emph{ball 
of radius~$r$ around~$v$}. The cost of a solution $C$ equivalently is the 
smallest value $\rho$ for which $\bigcup_{v\in C}B_v(\rho)=V$. The \kc problem 
has numerous applications in logistics where easily accessible locations need 
to 
be chosen in a network under a budget constraint. For instance, a budget may be 
available to build $k$ hospitals, shopping malls, or warehouses. These should 
be 
placed so that the distance from each point on the map to the closest facility 
is minimized.

The \kc problem is NP-hard~\cite{Vazirani01book}, and so \emph{approximation 
algorithms}~\cite{Vazirani01book,williamson2011design} as well as 
\emph{parameterized algorithms}~\cite{pc-book,downey2013fpt} have been 
developed 
for this problem. The former are algorithms that use polynomial time to compute 
an \emph{$\alpha$-approximation}, i.e., a solution that is at most $\alpha$ 
times worse than the optimum. For the latter, a \emph{parameter} $q$ is given 
as 
part of the input, and an optimum solution is computed in $f(q)\polyn$ time for 
some computable function $f$ independent of the input size $n$. The rationale 
behind such an algorithm is that it solves the problem efficiently in 
applications where the parameter is small. If such an algorithm exists, the 
corresponding  problem is called \emph{fixed-parameter tractable (FPT)} 
for~$q$. 
Another option is to consider \emph{parameterized approximation 
algorithms}~\cite{lokshtanov2017lossy,marx2008fpa}, which compute an 
$\alpha$-approximation in $f(q)\polyn$ time for some parameter~$q$.

By a result of Hochbaum and Shmoys~\cite{hochbaum1986bottleneck}, on general 
input graphs, a polynomial time $2$-approximation algorithm exists, and this 
approximation factor is also best possible, unless~P=NP. A natural parameter 
for 
\kc is the number of centers~$k$, for which however the problem is 
W[2]-hard~\cite{demaine2005planar-center}, and is thus unlikely to be FPT. In 
fact it is even W[2]-hard~\cite{DBLP:conf/icalp/Feldmann15} to compute a 
$(2-\eps)$-approximation for any~$\eps>0$, and thus parametrizing by $k$ does 
not help to overcome the polynomial-time inapproximability. For structural 
parameters such as the vertex-cover number or the feedback-vertex-set number 
the 
problem remains W[1]-hard~\cite{katsikarelis2017structural}, even when 
combining 
with the parameter $k$. For each of the two more general structural parameters 
treewidth and cliquewidth, an \emph{efficient parameterized approximation 
scheme 
(EPAS)} was shown to exist~\cite{katsikarelis2017structural}, i.e., a 
$(1+\eps)$-approximation can be computed in $f(\eps,w)\polyn$ time for any 
$\eps>0$, if $w$ is either the treewidth or the cliquewidth, and $n$ is the 
number of vertices.

Arguably however, graphs with low treewidth or cliquewidth do not model 
transportation networks well, since grid-like structures with large treewidth 
and cliquewidth can occur in road maps of big cities. As we focus on 
applications for \kc in logistics, here we consider more natural models for 
transportation networks. These include \emph{planar graphs}, \emph{low doubling 
metrics} such as the Euclidean or Manhattan plane, or the more recently studied 
\emph{low highway dimension graphs}. Our main result is that \kc is W[1]-hard 
on 
all of these graph classes \emph{combined}, even if adding $k$ and the 
pathwidth 
as parameters (note that the pathwidth is a stronger parameter than the 
treewidth). Before introducing these graph classes, let us formally state our 
theorem.

\begin{thm}\label{thm:hardness}
Even on planar graphs with edge lengths of doubling dimension $O(1)$, the \kc 
problem is \textup{W[1]}-hard for the combined parameter $(k,p,h)$, where $p$ 
is 
the pathwidth and $h$ the highway dimension of the input graph. Moreover, under 
ETH there is no $f(k,p,h)\cdot n^{o(p+\sqrt{k+h})}$ time 
algorithm\footnote{Here 
$o(p+\sqrt{k+h})$ means $g(p+\sqrt{k+h})$ for any function $g$ such that 
$g(x)\in o(x)$.} for the same restriction on the input graphs, for any 
computable function $f$.
\end{thm}

A \emph{planar graph} can be drawn in the plane without crossing edges. Such 
graphs constitute a realistic model for road networks, since overpasses and 
tunnels are relatively rare. It is known~\cite{plesnik1980planar-k-center} that 
also for planar graphs no $(2-\eps)$-approximation can be computed in 
polynomial 
time, unless P=NP. On the positive side, \kc is 
FPT~\cite{demaine2005planar-center} on unweighted planar graphs for the 
combined 
parameter $k$ and the optimum solution cost $\rho$. However, typically if $k$ 
is 
small then~$\rho$ is large and vice versa, and thus the applications for this 
combined parameter are rather limited. If the parameter is only $k$, then an 
$n^{O(\sqrt{k})}$ time algorithm exists for planar 
graphs~\cite{marx2015optimal}. By a very recent result~\cite{planar-EPTAS} the 
\kc problem on planar graphs with positive edge lengths admits an 
\emph{efficient polynomial-time bicriteria approximation scheme}, which for any 
$\eps>0$ in $f(\eps)\polyn$ time computes a solution that uses at most 
$(1+\eps)k$ centers and approximates the optimum with at most $k$ centers 
within 
a factor of $1+\eps$. This algorithm implies an EPAS for parameter~$k$ on 
planar graphs with edge lengths, since setting 
$\eps=\min\{\eps',\frac{1}{2k}\}$ 
forces the algorithm to compute a $(1+\eps')$-approximation in 
$f(k,\eps')\polyn$ time using at most $(1+\eps)k\leq k+\frac{1}{2}$ centers, 
i.e., at most $k$ centers as $k$ is an integer. This observation is 
complemented 
by our hardness result showing that it is necessary to approximate the solution 
when parametrizing by $k$ in planar graphs with edge lengths.

\begin{dfn}\label{dfn:dd}
The \emph{doubling dimension} of a metric $(X,\dist)$ is the smallest 
$d\in\mathbb{R}$ such that for any $r>0$, every ball of radius $2r$ is 
contained 
in the union of at most~$2^d$ balls of radius~$r$. The doubling dimension of a 
graph is the doubling dimension of its shortest-path metric.
\end{dfn}

Since a transportation network is embedded on a large sphere (namely the 
Earth), 
a reasonable model is to assume that the shortest-path metric abides to the 
Euclidean $L_2$-norm. In cities, where blocks of buildings form a grid of 
streets, it is reasonable to assume that the distances are given by the 
Manhattan $L_1$-norm. Every metric for which the distance function is given by 
the $L_q$-norm in $D$-dimensional space~$\mathbb{R}^D$ has doubling dimension 
$O(D)$. Thus a road network, which is embedded into $\mathbb{R}^2$ can 
reasonably be assumed to have constant doubling dimension. It is 
known~\cite{marx2005efficient} that \kc is W[1]-hard for parameter $k$ in 
two-dimensional Manhattan metrics. Also, no polynomial time 
$(2-\eps)$-approximation algorithm exists for \kc in two-dimensional Manhattan 
metrics~\cite{feder1988metric-k-center}, and no 
$(1.822-\eps)$\hy{}approximation 
for two-dimensional Euclidean metrics~\cite{feder1988metric-k-center}. On the 
positive side, Agarwal and Procopiuc~\cite{agarwal2002clustering} showed that 
for any $L_q$ metric in $D$ dimensions, the \kc problem can be solved optimally 
in $n^{O(k^{1-1/D})}$ time, and an EPAS exists for the combined parameter 
$(\eps,k,D)$. We generalize the latter to any metric of doubling dimension $d$, 
as formalized by the following theorem.

\begin{thm}\label{thm:alg}
Given a metric of doubling dimension $d$ and $\eps>0$, a 
$(1+\eps)$-approximation for \mbox{\kc} can be computed in 
$(k^k/\eps^{O(kd)})\polyn$ time.
\end{thm}

\cref{thm:hardness} complements this result by showing that it is necessary 
to 
approximate the cost of the solution if parametrizing by $k$ and $d$.

\begin{dfn}\label{dfn:hd}
The \emph{highway dimension} of a graph $G$ is the smallest $h\in\mathbb{N}$ 
such that, for some universal constant \mbox{$c\geq 4$}, for every $r\in 
\mathbb{R}^+$ and every ball $B_{cr}(v)$ of radius~$cr$, there is a set 
$H\subseteq B_{cr}(v)$ of \emph{hubs} such that $|H|\leq h$ and every shortest 
path of length more than~$r$ lying in $B_{cr}(v)$ contains a hub of $H$.
\end{dfn}

The highway dimension was introduced by Abraham et 
al.~\cite{abraham2010highway} 
as a formalization of the empirical observation by Bast et 
al.~\cite{bast2009ultrafast,bast2007transit} that in a road network, starting 
from any point $A$ and travelling to a sufficiently far point $B$ along the 
quickest route, one is bound to pass through some member of a sparse set of 
``access points'', i.e., the hubs. In contrast to planar and low doubling 
graphs, the highway dimension has the potential to model not only road networks 
but also more general transportation networks such as those given by 
air-traffic or public transportation. This is because in such networks longer 
connections tend to be serviced through larger and sparser stations, which act 
as hubs. Abraham et al.~\cite{abraham2010highway} were able to prove that 
certain shortest-path heuristics are provably faster in low highway dimension 
graphs than in general graphs. They specifically chose the constant $c=4$ in 
their original definition, but later work by Feldmann et 
al.~\cite{hd-embedding} showed that when choosing any constant $c>4$ in the 
definition, the structure of the resulting graphs can be exploited to obtain 
quasi-polynomial time approximation schemes for problems such as 
\pname{Travelling Salesman} or \pname{Facility Location}. Note that increasing 
the constant $c$ in \cref{dfn:hd} restricts the class of graphs further. 
Moreover, as shown by \citet[Section~9]{hd-embedding}, the highway dimension of 
a graph according to \cref{dfn:hd} can grow arbitrarily large by just a 
small change in the constant~$c$: for any $c$ there is a graph of highway 
dimension $1$ when using $c$ in \cref{dfn:hd}, which however has highway 
dimension $\Omega(n)$ for any constant larger than $c$.\footnote{We remark that 
these graphs have unbounded doubling dimension, and that an upper bound of 
$O(hc^d)$ on the highway dimension of any graph using constant $c$ in 
\cref{dfn:hd} can be shown, if the doubling dimension is $d$ and $h$ is the 
highway dimension using constant $4$.}
Other definitions of the highway dimension exist as 
well~\cite{abraham2010highway2,abraham2011vc, abraham2010highway} (see 
\citet[Section~9]{hd-embedding} and \citet{blum2019hierarchy} for detailed 
discussions).

Later, Becker et al.~\cite{hd-Becker2017} used the framework introduced by 
Feldmann et al.~\cite{hd-embedding} to show that whenever $c>4$ there is an 
EPAS 
for \kc parameterized by~$\eps$, $k$, and~$h$. Note that the highway dimension 
is always upper bounded by the vertex-cover number, as every edge of any 
non-trivial path is incident to a vertex cover. Hence the aforementioned 
W[1]-hardness result by Katsikarelis et al.~\cite{katsikarelis2017structural} 
for the combined parameter $k$ and the vertex-cover number proves that it is 
necessary to approximate the optimum when using $k$ and $h$ as the combined 
parameter. When parametrizing only by the highway dimension but not~$k$, it is 
not even known if a \emph{parameterized approximation scheme (PAS)} exists, 
i.e., an $f(\eps,h)\cdot n^{g(\eps)}$ time \mbox{$(1+\eps)$-}approximation 
algorithm for some computable functions $f,g$. However, under the 
\emph{Exponential Time Hypothesis (ETH)}~\cite{pc-book}, 
by~\cite{DBLP:conf/icalp/Feldmann15} there is no algorithm with doubly 
exponential $2^{2^{o(\sqrt{h})}}\polyn$ runtime computing a 
\mbox{$(2-\eps)$}\hy{}approximation for any $\eps>0$. The same 
paper~\cite{DBLP:conf/icalp/Feldmann15} also presents a $3/2$\hy{}approximation 
for \kc with runtime $2^{O(kh\log h)}\polyn$ for a more general definition of 
the highway dimension than the one given in \cref{dfn:hd} (based on 
so-called 
\emph{shortest path covers}). In contrast to the result of Becker et 
al.~\cite{hd-Becker2017}, it is not known whether a PAS exists when combining 
this more general definition of $h$ with $k$ as a parameter. 
\cref{thm:hardness} complements these results by showing that even on planar 
graphs of constant doubling dimension, for the combined parameter~$(k,h)$ no 
fixed-parameter algorithm exists, unless FPT=W[1]. Therefore approximating the 
optimum is necessary, regardless of whether~$h$ is according to 
\cref{dfn:hd} 
or the more general one from~\cite{DBLP:conf/icalp/Feldmann15}, and regardless 
of how restrictive \cref{dfn:hd} is made by increasing the constant~$c$.

\begin{dfn}\label{dfn:pw}
A \emph{path decomposition} of a graph $G=(V,E)$ is a path $P$ each of whose 
nodes~$v$ is labelled by a bag $K_v \subseteq V$ of vertices of $G$, and has 
the following properties:
 \begin{enumerate}[(a)]
\item\label{item:tw-union} $\bigcup_{v \in V(P)} K_v = V$, 
\item\label{item:tw-edges} for every edge $\{u,w\}\in E$ there is a
  node $v \in V(P)$ such that~$K_v$ contains both $u$ and~$w$,
\item\label{item:tw-vertices} for every $v\in V$ the set $\{u \in V(P) \mid v 
\in K_u\}$ induces a connected subpath of $P$.
 \end{enumerate}
The \emph{width} of the path decomposition is $\max\{|K_v|-1\mid v \in V(P)\}$. 
The \emph{pathwidth} $p$ of a graph~$G$ is the minimum width 
among all path decompositions for~$G$.
\end{dfn}

The pathwidth of a graph is always at least as large as its treewidth (for 
which 
the path $P$ in the above definition is replaced by a tree). Thus, as mentioned 
above, arguably, bounded pathwidth graphs are not a good model for 
transportation networks. Also it is already known that \kc is W[1]-hard for 
this 
parameter, even when combining it with $k$~\cite{katsikarelis2017structural}. 
We 
include this well-studied parameter here nonetheless, since the reduction of 
our 
hardness result in \cref{thm:hardness} implies that \kc is W[1]\hy{}hard 
even 
for planar graphs with edge lengths when \emph{combining} any of the parameters 
$k$, $h$, $d$, and $p$. As noted by \citet{hd-embedding} and 
\citet{blum2019hierarchy}, these parameters are not bounded in terms of each 
other, i.e., they are incomparable. Furthermore, the doubling dimension is in 
fact bounded by a constant in \cref{thm:hardness}. Hence, even if one were 
to 
combine all the models presented above and assume that a transportation network 
is planar, is embeddable into some metric of constant doubling dimension, has 
bounded highway dimension, and even has bounded pathwidth, the \kc problem 
cannot be solved efficiently, unless FPT=W[1]. Thus it seems unavoidable to 
approximate the problem in transportation networks, when developing fast 
algorithms.

\subsection{Related work}

The above mentioned efficient bicriteria approximation 
scheme~\cite{planar-EPTAS} improves on a previous (non-efficient) bicriteria 
approximation scheme~\cite{eisenstat2014planar}, which for any $\eps>0$ and 
planar input graph with edge lengths computes a $(1+\eps)$-approximation with 
at 
most $(1+\eps)k$ centers in time $n^{f(\eps)}$ for some function $f$ (note that 
in contrast to above, such an algorithm does not imply a PAS for parameter 
$k$). 
The paper by Demaine et al.~\cite{demaine2005planar-center} on the \kc problem 
in unweighted planar graphs also considers the so-called class of \emph{map 
graphs}, which is a superclass of planar graphs that is not minor-closed. They 
show that the problem is FPT on unweighted map graphs for the combined 
parameter 
$(k,\rho)$. Also for the tree-depth, \kc is 
FPT~\cite{katsikarelis2017structural}. Another parameter related to 
transportation networks is the \emph{skeleton dimension}, for which it was 
recently shown~\cite{blum2019hierarchy} that, under ETH, no 
$2^{2^{o(\sqrt{s})}}\polyn$ time algorithm can compute a 
$(2-\eps)$-approximation for any $\eps>0$, if the skeleton dimension is $s$. It 
is not known whether this parameter yields any approximation schemes when 
combined with for instance $k$, as is the case for the highway dimension.

A closely related problem to \kc is the \pname{$\rho$-Dominating Set} problem, 
in which $\rho$ is given and the number $k$ of centers covering a given graph 
with $k$ balls of radius $\rho$ needs to be minimized. As this generalizes the 
\pname{Dominating Set} problem, no $(\ln(n)-\eps)$\hy{}approximation is 
possible 
in polynomial time~\cite{dinur2014setcover}, unless P=NP, and computing an 
$f(k)$-approximation is W[1]-hard~\cite{dom-set} when parametrizing by~$k$, for 
any computable function~$f$.

\section{The reduction}
\label{sec:reduction}

In this section we give a reduction from the \pname{Grid Tiling with 
Inequality} 
(\gti) problem, which was introduced by \citet{marx2014limited} and is defined 
as follows. Given $\kappa^2$ non-empty sets $S_{i,j}\subseteq [n]^2$ of pairs 
of integers,\footnote{For any positive integer $q$, throughout this article 
$[q]$ means $\{1,\ldots,q\}$.} where $i,j\in[\kappa]$, the task is to select 
one pair $s_{i,j}\in S_{i,j}$ for each set such that
\begin{itemize}
 \item if $s_{i,j}=(a,b)$ and $s_{i+1,j}=(a',b')$ for $i\leq \kappa-1$ then 
$a\leq a'$, and
 \item if $s_{i,j}=(a,b)$ and $s_{i,j+1}=(a',b')$ for $j\leq \kappa-1$ then 
$b\leq b'$.
\end{itemize}
The \gti problem is W[1]-hard~\cite{pc-book} for parameter $\kappa$, and 
moreover, under ETH has no $f(\kappa)\cdot n^{o(\kappa)}$ time algorithm for 
any 
computable function $f$.

\subsection{Construction}

Given an instance $\mc{I}$ of \gti with $\kappa^2$ sets, we construct the 
following graph $G_\mc{I}$. First, for each set $S_{i,j}$, where $1\leq i,j\leq 
\kappa$, we fix an arbitrary order on its elements, so that 
$S_{i,j}=\{s_1,\ldots,s_\sigma\}$, where $\sigma\leq n^2$. We then construct a 
gadget $G_{i,j}$ for~$S_{i,j}$, which contains a cycle~$O_{i,j}$ of length 
$16n^2+4$ in which each edge has length~$1$ (see \cref{fig:red}(a)). 
Additionally we introduce five vertices $x^1_{i,j}$, $x^2_{i,j}$, $x^3_{i,j}$, 
$x^4_{i,j}$, and~$y_{i,j}$. If $O_{i,j}=(v_1,v_2,\ldots,v_{16n^2+4},v_1)$ then 
we connect these five vertices to the cycle as follows. The vertex $y_{i,j}$ is 
adjacent to the four vertices $v_{1}$, $v_{4n^2+2}$, $v_{8n^2+3}$, 
and~$v_{12n^2+4}$, with edges of length $2n^2+1$ each. For every 
$\tau\in[\sigma]$ 
and $s_\tau\in S_{i,j}$, if $s_\tau=(a,b)$ we add the four edges
\begin{itemize}
 \item $x^1_{i,j} v_{\tau}$ of length $\ell'_a=2n^2-\frac{a}{n+1}$,
 \item $x^2_{i,j} v_{\tau+4n^2+1}$ of length $\ell_b=2n^2+\frac{b}{n+1}-1$,
 \item $x^3_{i,j} v_{\tau+8n^2+2}$ of length $\ell_a=2n^2+\frac{a}{n+1}-1$, and
 \item $x^4_{i,j} v_{\tau+12n^2+3}$ of length $\ell'_b=2n^2-\frac{b}{n+1}$.
\end{itemize}
We say that the element $s_\tau$ corresponds to the four vertices $v_{\tau}$, 
$v_{\tau+4n^2+1}$, $v_{\tau+8n^2+2}$, and~$v_{\tau+12n^2+3}$. Note that $s_1$ 
(which always exists) corresponds to the four vertices adjacent to $y_{i,j}$. 
Note also that $2n^2-1 < \ell_a,\ell'_a,\ell_b,\ell'_b<2n^2$, since $a,b\in 
[n]$.

The gadgets $G_{i,j}$ are now connected to each other in a grid-like fashion 
(see \cref{fig:red}(b)). That is, for $j\leq\kappa-1$ we add a path 
$P_{i,j}$ 
between $x^2_{i,j}$ and $x^4_{i,j+1}$ with $n+2$ edges of length 
$\frac{1}{n+2}$ 
each. Analogously, for \mbox{$i\leq\kappa-1$} we introduce a path $P'_{i,j}$ 
between $x^3_{i,j}$ and $x^1_{i+1,j}$ that has $n+2$ edges, each of 
length~$\frac{1}{n+2}$. Note that these paths all have length~$1$.

\begin{figure*}[t!]
    \centering
    \begin{subfigure}[t]{0.5\textwidth}
    \centering
 \includegraphics[width=\linewidth]{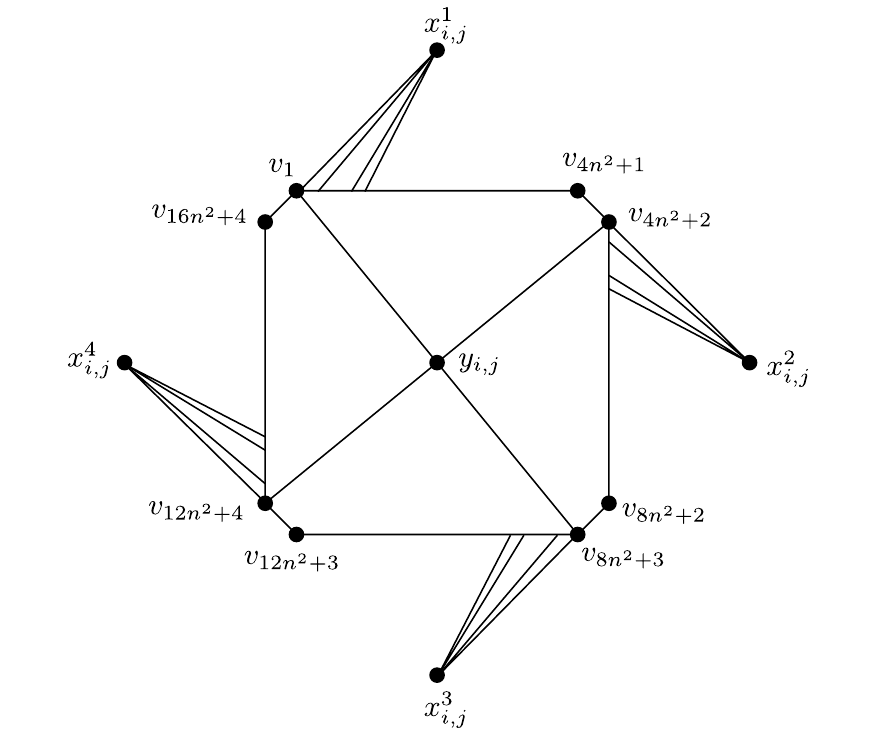}
        \caption{The gadget $G_{i,j}$ in the reduction.}
    \end{subfigure}%
    ~ 
    \begin{subfigure}[t]{0.5\textwidth}
    \centering
 \includegraphics[width=\linewidth]{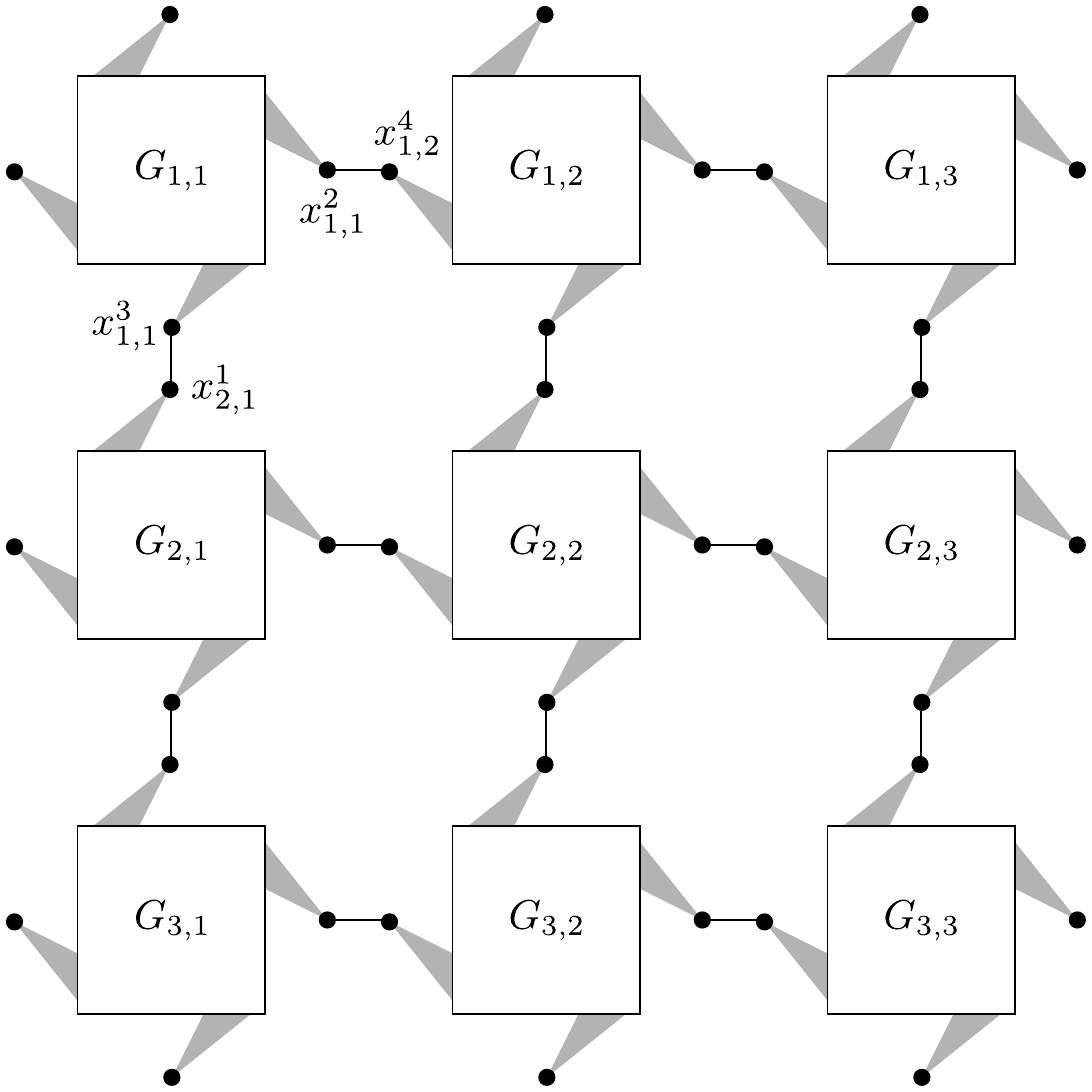}
        \caption{Connecting the gadgets in a grid-like fashion.}
    \end{subfigure}
    \caption{The structure of the graph $G_\mc{I}$ constructed for the 
reduction.}\label{fig:red}
\end{figure*}

The resulting graph $G_\mc{I}$ forms an instance of \kc with $k=5\kappa^2$. We 
claim that the instance $\mc{I}$ of \gti has a solution if and only if the 
optimum solution to \kc on $G_\mc{I}$ has cost at most $2n^2$. We note at this 
point that the reduction would still work when removing the vertices $y_{i,j}$ 
and decreasing $k$ to~$4\kappa^2$. However their existence will greatly 
simplify analysing the doubling dimension of $G_\mc{I}$ in \cref{sec:prop}.

\subsection{A solution to the \gti instance implies a \kc instance with cost 
$2n^2$}

Recall that we fixed an order of each set $S_{i,j}$, so that each element 
$s_\tau\in S_{i,j}$ corresponds to four equidistant vertices on cycle $O_{i,j}$ 
with distance $4n^2+1$ between consecutive such vertices on the cycle. If 
$s_\tau\in S_{i,j}$ is in the solution to the \gti instance $\mc{I}$, let 
$C_{i,j}=\{v_{\tau}, v_{\tau+4n^2+1}, v_{\tau+8n^2+2}, v_{\tau+12n^2+3}, 
y_{i,j}\}$ contain the vertices of $O_{i,j}$ corresponding to $s_\tau$ in 
addition to $y_{i,j}$. The solution to the \kc instance $G_\mc{I}$ is given by 
the union $\bigcup_{i,j\in[\kappa]} C_{i,j}$, which consists of exactly 
$5\kappa^2$ centers in total.

Let us denote the set containing the four vertices of $C_{i,j}\cap V(O_{i,j})$ 
by $C_{i,j}^O$, and note that each of these four vertices covers $4n^2+1$ 
vertices of $O_{i,j}$ with balls of radius $2n^2$, as each edge of $O_{i,j}$ 
has 
length~$1$. Since the distance between any pair of centers in $C_{i,j}^O$ is at 
least $4n^2+1$, these four sets of covered vertices are pairwise disjoint. Thus 
the total number of vertices covered by $C_{i,j}^O$ on $O_{i,j}$ is $16n^2+4$, 
i.e., all vertices of the cycle $O_{i,j}$ are covered. Recall that the lengths 
of the edges between the vertices $x^1_{i,j}$, $x^2_{i,j}$, $x^3_{i,j}$, and 
$x^4_{i,j}$ and the cycle $O_{i,j}$ are $\ell_a,\ell'_a,\ell_b,\ell'_b<2n^2$. 
Hence the centers in $C_{i,j}^O$ also cover $x^1_{i,j}$, $x^2_{i,j}$, 
$x^3_{i,j}$, and $x^4_{i,j}$ by balls of radius $2n^2$.

Now consider a path connecting two neighbouring gadgets, e.g., $P_{i,j}$ 
connecting $x^2_{i,j}$ and $x^4_{i,j+1}$. The center sets $C_{i,j}^O$ and 
$C_{i,j+1}^O$ contain vertices corresponding to the respective elements $s\in 
S_{i,j}$ and $s'\in S_{i,j+1}$ of the solution to the \gti instance. This means 
that if $s=(a,b)$ and $s'=(a',b')$ then $b\leq b'$. Thus the closest centers of 
$C_{i,j}^O$ and $C_{i,j+1}^O$ are at distance $\ell_b+1+\ell'_{b'}$ from each 
other, as $P_{i,j}$ has length $1$. From $b\leq b'$ we get
\[
 \ell_b+1+\ell'_{b'}= 2n^2+\frac{b}{n+1}-1+1+2n^2-\frac{b'}{n+1}\leq 4n^2.
\]
Therefore all vertices of $P_{i,j}$ are covered by the balls of radius $2n^2$ 
around the two closest centers of $C_{i,j}^O$ and $C_{i,j+1}^O$. Analogously, 
we 
can also conclude that any path $P'_{i,j}$ connecting some vertices $x^1_{i,j}$ 
and $x^3_{i+1,j}$ is covered, using the fact that if $(a,b)\in S_{i,j}$ and 
$(a',b')\in S_{i+1,j}$ are in the solution to the \gti instance then~$a\leq a'$.

Finally, the remaining center vertices in $\bigcup_{i,j\in[\kappa]} 
C_{i,j}\setminus C_{i,j}^O$ cover the additional vertex~$y_{i,j}$ in each 
gadget~$G_{i,j}$.

\subsection{A \kc instance with cost $2n^2$ implies a solution to the \gti 
instance}

Each vertex $y_{i,j}$ must be contained in any solution of cost at most $2n^2$, 
since the distance from $y_{i,j}$ to any other vertex is more than $2n^2$. This 
already uses $\kappa^2$ of the available $5\kappa^2$ centers.

We now prove that in any solution to the \kc instance $G_\mc{I}$ of cost at 
most~$2n^2$, each cycle $O_{i,j}$ must contain exactly four centers. Recall 
that 
$\ell_a,\ell'_a,\ell_b,\ell'_b>2n^2-1$, that $y_{i,j}$ is incident to four 
edges 
of length $2n^2+1$ each, and that each edge of $O_{i,j}$ has length $1$. Now 
consider the vertices $v_{4n^2+1}$, $v_{8n^2+2}$, $v_{12n^2+3}$, 
and~$v_{16n^2+4}$, each of which is not connected by an edge to any vertex 
$x^q_{i,j}$, where $q\in[4]$, nor to~$y_{i,j}$. Thus each of these four 
vertices 
must be covered by centers on the cycle~$O_{i,j}$ if the radius of each ball is 
at most $2n^2$. Furthermore, the distance between each pair of these four 
vertices is at least $4n^2+1$, which means that any solution of cost at most 
$2n^2$ needs at least four centers on $O_{i,j}$ to cover these four vertices. 
Since there are $\kappa^2$ cycles and only $4\kappa^2$ remaining available 
centers, we proved that each cycle~$O_{i,j}$ contains exactly four centers, and 
apart from the $y_{i,j}$ vertices no other centers exist in the graph 
$G_\mc{I}$. 

Let $C_{i,j}^O$ be the set of four centers contained in $O_{i,j}$. As each 
center of $C_{i,j}^O$ covers at most $4n^2+1$ vertices of $O_{i,j}$ by balls of 
radius at most $2n^2$, to cover all $16n^2+4$ vertices of $O_{i,j}$ these four 
centers must be equidistant with distance exactly $4n^2+1$ between consecutive 
centers on $O_{i,j}$. Furthermore, since $\ell_a,\ell'_a,\ell_b,\ell'_b>2n^2-1$ 
and each edge of $O_{i,j}$ has length~$1$, to cover~$x^q_{i,j}$ for any 
$q\in[4]$ some center of $C_{i,j}^O$ must lie on a vertex of $O_{i,j}$ adjacent 
to~$x^q_{i,j}$. This means that the four centers of $C_{i,j}^O$ are exactly 
those vertices $v_{\tau+(q-1)(4n^2+1)}$ corresponding to element $s_\tau$ of 
$S_{i,j}$.

It remains to show that the elements corresponding to the centers in 
$\bigcup_{i,j\in[\kappa]} C^O_{i,j}$ form a solution to the \gti instance 
$\mc{I}$. For this, consider two neighbouring gadgets $G_{i,j}$ and 
$G_{i,j+1}$, 
and let $(a,b)\in S_{i,j}$ and $(a',b')\in S_{i,j+1}$ be the respective 
elements 
corresponding to the center sets $C_{i,j}^O$ and $C_{i,j+1}^O$. Note that for 
any $\hat b\in[n]$ we have $\ell_b\leq \ell_{\hat b}+1$ and $\ell'_{b'}\leq 
\ell'_{\hat b}+1$. Since every edge of the cycles $O_{i,j}$ and $O_{i,j+1}$ has 
length~$1$, this means that the distance from the closest centers $v\in 
C_{i,j}^O$ and $v'\in C_{i,j+1}^O$ to $x^2_{i,j}$ and $x^4_{i,j+1}$, 
respectively, is determined by the edges of length $\ell_b$ and $\ell'_{b'}$ 
incident to $v$ and $v'$, respectively. In particular, the distance between $v$ 
and~$v'$ is $\ell_b+1+\ell'_{b'}$, as the path $P_{i,j}$ connecting $x^2_{i,j}$ 
and $x^4_{i,j+1}$ has length $1$. Assume now that $b>b'$, which means that 
$b\geq b'+1$ since $b$ and $b'$ are integer. Hence this distance is
\[
\ell_b+1+\ell'_{b'}= 2n^2+\frac{b}{n+1}-1+1+2n^2-\frac{b'}{n+1} \geq 
4n^2+\frac{1}{n+1}.
\]
As the centers $v$ and $v'$ only cover vertices at distance at most $2n^2$ 
each, 
while the edges of the path $P_{i,j}$ have length 
$\frac{1}{n+2}<\frac{1}{n+1}$, 
there must be some vertex of $P_{i,j}$ that is not covered by the center set. 
However this contradicts the fact that the centers form a feasible solution 
with 
cost at most $2n^2$, and so $b\leq b'$.

An analogous argument can be made for neighbouring gadgets $G_{i,j}$ and 
$G_{i+1,j}$, so that $a\leq a'$ for the elements $(a,b)\in S_{i,j}$ and 
$(a',b')\in S_{i+1,j}$ corresponding to the centers in $C_{i,j}^O$ 
and~$C_{i+1,j}^O$, respectively. Thus a solution to $G_\mc{I}$ of cost at 
most~$2n^2$ implies a solution to $\mc{I}$.

\section{Properties of the constructed graph}
\label{sec:prop}

The reduction of \cref{sec:reduction} proves that the \kc problem is 
W[1]-hard for parameter $k$, since the reduction can be done in polynomial time 
and $k$ is a function of $\kappa$. Since this function is quadratic, we can 
also 
conclude that, under ETH, there is no $f(k)\cdot n^{o(\sqrt{k})}$ time 
algorithm 
for \kc. We will now show that the graph constructed in the reduction has 
various additional properties from which we will be able to conclude 
\cref{thm:hardness}. First off, it is easy to see that any constructed 
graph~$G_\mc{I}$ for an instance $\mc{I}$ of \gti is planar 
(cf.~\cref{fig:red}). We go on to prove that $G_\mc{I}$ has constant 
doubling dimension.

\begin{lem}\label{lem:dd}
The graph $G_\mc{I}$ has doubling dimension at most \mbox{$\log_2(324)\approx 
8.34$} for $n\geq 2$.
\end{lem}
\begin{proof}
planar. 
To bound the doubling dimension of the graph $G_\mc{I}$, consider the 
shortest-path metric on the vertex set $Y=\{y_{i,j}\in V(G_\mc{I})\mid 
i,j\in[\kappa]\}$ given by the distances between these vertices in~$G_\mc{I}$. 
On an intuitive level, as these vertices are arranged in a grid-like fashion, 
the shortest-path metric on $Y$ approximates the $L_1$\hy{}metric. We consider 
a set of index pairs, for which the corresponding vertices in $Y$ roughly 
resemble a ball in the shortest-path metric on $Y$. That is, for any 
$a\in\mathbb{N}_0$ consider the set of index pairs 
$A_{i,j}(a)=\{(i',j')\in[\kappa]^2\mid |i-i'|+|j-j'|\leq a\}$, and let 
$V_{i,j}(a)\subseteq V(G_\mc{I})$ contain all vertices of gadgets $G_{i',j'}$ 
such that $(i',j')\in A_{i,j}(a)$ in addition to the vertices of paths of 
length 
$1$ connecting these gadgets to each other and to any adjacent gadgets 
$G_{i'',j''}$ such that $(i'',j'')\notin A_{i,j}(a)$. We call the vertices 
$x^q_{i'',j''}\in V_{i,j}(a)$ such that $(i'',j'')\notin A_{i,j}(a)$, i.e., the 
endpoints of the latter paths of length $1$, the \emph{boundary vertices of 
$V_{i,j}(a)$}. We consider $y_{i,j}$ as the center of $V_{i,j}(a)$. We would 
like to determine the smallest radius of a ball around $y_{i,j}$ that contains 
all of $V_{i,j}(a)$, and the largest radius of a ball around $y_{i,j}$ that is 
entirely contained in $V_{i,j}(a)$. For this we need the following claim, which 
we will also reuse later.

\begin{claim}\label{clm:dists}
For any gadget $G_{i,j}$ and $q,q'\in[4]$ with $q\neq q'$, the distance between 
$x^q_{i,j}$ and~$x^{q'}_{i,j}$ in $G_\mc{I}$ lies between $7n^2-1$ and $8n^2+2$.
\end{claim}
\begin{proof}
The distance between $x^q_{i,j}$ and $x^{q'}_{i,j}$ is less than 
$2(2n^2+2n^2+1)=8n^2+2$, via the path passing through $y_{i,j}$ and the two 
vertices of $O_{i,j}$ adjacent to $y_{i,j}$, $x^q_{i,j}$, and $x^{q'}_{i,j}$. 
Note that the shortest path between $x^q_{i,j}$ and $x^{q'}_{i,j}$ inside the 
gadget $G_{i,j}$ does not necessarily pass through $y_{i,j}$, but may pass 
along the cycle $O_{i,j}$ instead. This is because the set $S_{i,j}$ of the 
\gti instance may contain up to $n^2$ elements, which would imply a direct edge 
from $x^{q}_{i,j}$ to $v_{n^2+(q-1)(4n^2+1)}$ on $O_{i,j}$. Thus we can give a 
lower bound of $2(2n^2-1)+3n^2+1=7n^2-1$ for the distance between~$x^{q}_{i,j}$ 
and $x^{q'}_{i,j}$ inside of $G_{i,j}$. This is also the shortest path between 
these vertices in~$G_\mc{I}$, since any other path needs to pass through at 
least three gadgets.
\cqed
\end{proof}

We define the \emph{circumradius} of $V_{i,j}(a)$ as the maximum distance 
from $y_{i,j}$ to any vertex inside of~$V_{i,j}(a)$, while the \emph{inradius} 
of $V_{i,j}(a)$ is the minimum distance from $y_{i,j}$ to any vertex 
outside of~$V_{i,j}(a)$. Note that $V_{i,j}(a)\subseteq B_{y_{i,j}}(r)$ if~$r$ 
is the circumradius, and $B_{y_{i,j}}(r-\eps)\subseteq V_{i,j}(a)$ for any 
$\eps>0$ if $r$ is the inradius. Any shortest path from $y_{i,j}$ to a 
vertex in~$V_{i,j}(a)$ passes through the gadget~$G_{i,j}$, at most $a$ 
additional gadgets $G_{i',j'}$ with $(i',j')\in A_{i,j}(a)$, the paths of 
length~$1$ connecting these gadgets, and possibly one path of length $1$ to 
reach a boundary vertex of $V_{i,j}(a)$. Hence, by \cref{clm:dists}, the 
circumradius of~$V_{i,j}(a)$ is less than $(8n^2+2)a+(a+1)+4n^2+1 = 
(8n^2+3)a+4n^2+2$, since the distance from~$y_{i,j}$ to any $x^q_{i,j}$ is less 
than $4n^2+1$. To reach any vertex outside of $V_{i,j}(a)$ from $y_{i,j}$ it is 
necessary to first reach $x^q_{i,j}$ for some $q\in[4]$, then pass through 
$a$ gadgets~$G_{i',j'}$ with $(i',j')\in A_{i,j}(a)$, in addition to $a$ paths 
of length~$1$ connecting them and~$G_{i,j}$, and finally pass through another 
path of length $1$ to reach a boundary vertex of $V_{i,j}(a)$. From the 
boundary, a vertex not in $V_{i,j}(a)$ can be reached on some cycle 
$O_{i'',j''}$ with $(i'',j'')\notin A_{i,j}(a)$. The distance from~$y_{i,j}$ 
to~$x^q_{i,j}$ is more than~$4n^2$ and the distance from a boundary vertex 
$x^{q'}_{i'',j''}$ to any vertex of $O_{i'',j''}$ is more than $2n^2-1$. Hence, 
by \cref{clm:dists}, the inradius of~$V_{i,j}(a)$ is more than 
$(7n^2-1)a+(a+1)+4n^2+(2n^2-1)=7n^2 a+6n^2$.

Now consider any ball $B_v(2r)$ of radius $2r$ around some vertex $v$ in 
$G_\mc{I}$ for which we need to bound the number of balls of half the radius 
with which to cover~$B_v(2r)$. Let $y_{i,j}$ be the closest vertex of~$Y$ 
to~$v$. The distance between $y_{i,j}$ and $v$ is at most $2(2n^2+1)=4n^2+2$, 
whether $v$ lies on~$O_{i,j}$ or on one of the paths of length $1$ connecting 
$G_{i,j}$ with an adjacent gadget. Hence the ball $B_v(2r)$ is contained in a 
ball of radius $4n^2+2+2r$ around~$y_{i,j}$. The latter ball is in turn 
contained in the set $V_{i,j}(a)$ centered at $y_{i,j}$ if the ball's radius is 
less than the inradius of~$V_{i,j}(a)$. This in particular happens if 
$4n^2+2+2r\leq 7n^2a+6n^2$, which for instance is true if 
$a=\lceil\frac{2+2r-2n^2}{7n^2}\rceil$. Assume first that $r\geq 12n^2+5$, 
which 
implies that $a>0$ and so $V_{i,j}(a)$ is well-defined.

At the same time, any set $V_{i',j'}(a')$ is contained in a ball of radius $r$ 
around~$y_{i',j'}$ if its circumradius is at most $r$, i.e., 
$(8n^2+3)a'+4n^2+2\leq r$. This is for instance true if 
$a'=\lfloor\frac{r-4n^2-2}{8n^2+3}\rfloor$. Note that $r\geq 12n^2+5$ means 
that $a'\geq 0$ and so $V_{i',j'}(a')$ is well-defined. We may cover all 
vertices of $A_{i,j}(a)$ with $\lceil\frac{2a+1}{2a'+1}\rceil^2$ sets 
$A_{i',j'}(a')$, since in $Y$ these sets correspond to ``squares rotated by 45 
degrees'' (i.e., balls in~$L_1$) of diameter $2a+1$ and $2a'+1$, respectively. 
Thus we can cover $V_{i,j}(a)$ with $\lceil\frac{2a+1}{2a'+1}\rceil^2$ sets 
$V_{i',j'}(a')$, i.e., we can cover a ball of radius $2r$ in $G_\mc{I}$ with 
\[
\left\lceil\frac{2a+1}{2a'+1}\right\rceil^2 \leq 
\left\lceil\frac{2(\frac{2+2r-2n^2}{7n^2})+3}{2(\frac{r-4n^2-2}{8n^2+3})-1}
\right\rceil^2 = 
\left\lceil\frac{(8n^2+3)(4+4r+17n^2)}{7n^2(2r-7-16n^2)}\right\rceil^2 \leq
\left\lceil\frac{9(65r-37)}{7(8r-4)}\right\rceil^2 \leq 121
\] 
balls of radius $r$, using that $r\geq 12n^2+5$ implies $n^2\leq (r-5)/12$.

Next consider the case when $2n^2+1\leq r<12n^2+5$. We know from above that 
$B_v(2r)$ is contained in $V_{i,j}(a)$ if 
$a=\lceil\frac{2+2r-2n^2}{7n^2}\rceil$, which is well-defined as $r\geq 2n^2+1$ 
implies $a\geq 0$. Using $r<12n^2+5$ and $n\geq 2$ we get $a\leq 4$. The set 
$V_{i,j}(4)$ contains at most $(2\cdot 4+1)^2=81$ gadgets. On each of the 
cycles~$O_{i',j'}$ with $(i',j')\in A_{i,j}(4)$ we may choose the four vertices 
$v_{1}$, $v_{4n^2+2}$, $v_{8n^2+3}$, and~$v_{12n^2+4}$ adjacent to $y_{i',j'}$ 
as centers for balls of radius $r$. Note that as $r\geq 2n^2+1$, the vertex 
$y_{i',j'}$, every vertex of the cycle $O_{i',j'}$, and also every vertex on a 
path of length $1$ adjacent to gadget $G_{i',j'}$ is at distance at most $r$ to 
one of these four vertices. Thus at most $4\cdot 81=324$ balls of half the 
radius are needed to cover all vertices of~$B_v(2r)$.

The next case we consider is $n^2-1\leq r < 2n^2+1$. Again, $B_v(2r)$ is 
contained in~$V_{i,j}(a)$ if $a=\lceil\frac{2+2r-2n^2}{7n^2}\rceil$, which is 
well-defined as $r\geq n^2-1$ implies $a\geq 0$. Using $r < 2n^2+1$ and $n\geq 
1$ we obtain $a\leq 1$, which in turn means that the number of gadgets in 
$V_{i,j}(a)$ now is at most $(2a+1)^2\leq 9$. To cover a cycle $O_{i',j'}$ with 
$(i',j')\in A_{i,j}(1)$, we may choose centers for balls of radius $r$ 
equidistantly at every~$\lfloor 2r\rfloor$-th vertex of $O_{i',j'}$, as all 
edges of the cycle have length $1$. We may lower bound $\lfloor 2r\rfloor\geq 
2n^2-3$  using $r\geq n^2-1$. Since every cycle contains $16n^2+4$ vertices, 
the 
number of balls to cover a cycle is at most $16n^2+4/\lfloor 2r\rfloor \leq 
68/5 
\leq 14$, using the previous bound and $n\geq 2$. Hence at most $14\cdot 9=126$ 
balls of half the radius are needed to cover the cycles in $B_v(2r)$. We can 
then cover the $9$ vertices $y_{i',j'}$ where $(i',j')\in A_{i,j}(1)$ and the 
$2\cdot 9+2\cdot 3=24$ paths of length $1$ contained in $V_{i,j}(a)$ with a 
ball 
of radius~$r$ each, as $r\geq n^2-1\geq 1$ using $n\geq 2$. Hence a total of at 
most $9+24+126=159$ balls of half the radius suffice to cover $B_v(2r)$.

Finally, if $r<n^2-1$, then a ball $B_v(2r)$ contains only a subpath of some 
cycle~$O_{i,j}$, a subpath of a path of length $1$ connecting two gadgets, or a 
single vertex $y_{i,j}$, since any edge connecting a cycle $O_{i,j}$ to 
$y_{i,j}$ or some $x^q_{i,j}$ has length more than $2n^2-1>2(n^2-1)> 2r$. In 
this case at most~$3$ balls of radius $r$ suffice to cover all vertices 
of~$B_v(2r)$.
\end{proof}

We next show that we can bound the parameters $p$ and $h$, i.e., the pathwidth 
and highway dimension of~$G_\mc{I}$, linearly by $\kappa$ and 
$k=\Theta(\kappa^2)$, respectively. Note that the following lemma bounds the 
highway dimension in terms of $k$, no matter how restrictive we make 
\cref{dfn:hd} by increasing the constant $c$.

\begin{lem}\label{lem:hd}
For any constant $c$ of \cref{dfn:hd}, the graph $G_\mc{I}$ has highway 
dimension at most~$O(\kappa^2)$.
\end{lem}
\begin{proof}
For any scale $r\in\mathbb{R}^+$ and universal constant $c\geq 4$ we will 
define 
a hub set $H_r\subseteq V$ hitting all shortest paths of length more than $r$ 
in 
$G_\mc{I}$, such that $|H_r\cap B_{cr}(v)|=O(\kappa^2)$ for any ball 
$B_{cr}(v)$ 
of radius $cr$ in $G_\mc{I}$. This bounds the highway dimension to 
$O(\kappa^2)$ 
according to \cref{dfn:hd}.

Let $X=\{y_{i,j},x^q_{i,j}\mid q\in[4] \text{ and } i,j\in[\kappa]\}$ so that 
it 
contains all vertices connecting gadgets $G_{i,j}$ to each other in addition to 
the vertices $y_{i,j}$. If $r>8n^2+2$ then $H_r=X$. Any shortest path 
containing 
only vertices of a cycle $O_{i,j}$ has length at most $8n^2+2$, since the cycle 
has length $16n^2+4$. Any (shortest) path that is a subpath of a path 
connecting 
two gadgets has length at most $1$. Hence any shortest path of length more than 
$8n^2+2$ must contain some vertex of $X$. The total size of~$X$ is $5\kappa^2$, 
and so any ball, no matter its radius, also contains at most this many hubs of 
$H_r$.

If $1\leq r\leq 8n^2+2$ then any path of length more than $r$ but not 
containing 
any vertex of~$X$ must lie on some cycle 
$O_{i,j}=(v_1,v_2,\ldots,v_{16n^2+4},v_1)$. We define the set 
$W^{i,j}_r=\{v_{1+\lambda\lfloor r\rfloor}\in V(O_{i,j}) \mid 
\lambda\in\mathbb{N}_0\}$, i.e., it contains every $r$-th vertex on the cycle 
after rounding down. This means that every path on $O_{i,j}$ of length more 
than 
$r$ contains a vertex of~$W^{i,j}_r$. Thus for these values of $r$ we set 
$H_r=X\cup\bigcup_{i,j\in[\kappa]} W^{i,j}_r$. Any ball $B_{cr}(v)$ of 
radius~$cr$ contains $O(c)$ hubs of any~$W^{i,j}_r$. By \cref{clm:dists}, 
the 
distance between any pair of the four vertices~$x^q_{i,j}$, where $q\in[4]$, 
that connect a gadget $G_{i,j}$ to other gadgets, is more than~$7n^2-1$. This 
means that $B_{cr}(v)$ can only intersect $O(c^2)$ gadgets, since $cr\leq 
c(8n^2+2)\leq 2c(7n^2-1)$ if $n\geq1$ and the gadgets are connected in a 
grid-like fashion. Hence the ball $B_{cr}(v)$ only contains~$O(c)$ hubs for 
each 
of the $O(c^2)$ sets $W^{i,j}_r$ for which $B_{cr}(v)$ intersect the respective 
gadget~$G_{i,j}$. At the same time each gadget contains only~$5$ vertices of 
$X$. Thus if $c$ is a constant, then the number of hubs of $H_r$ in $B_{cr}(v)$ 
is constant.

If $r<1$, a path of length more than $r$ may be a subpath of a path connecting 
two gadgets. Recall that the paths $P_{i,j}$ connecting $x^2_{i,j}$ and 
$x^4_{i,j+1}$ for $j\leq\kappa-1$, and the paths $P'_{i,j}$ connecting 
$x^3_{i,j}$ and $x^1_{i+1,j}$ for $i\leq\kappa-1$, consist of $n+2$ edges of 
length~$\frac{1}{n+2}$ each. If $P_{i,j}=(u_0,u_1,\ldots,u_{n+2})$, we define 
the set $U^{i,j}_r=\{u_{\lambda\lfloor r(n+2)\rfloor}\in V(P_{i,j})\mid 
\lambda\in\mathbb{N}_0\}$, and if $P'_{i,j}=(u_0,u_1,\ldots,u_{n+2})$, we 
define 
the set $\tilde U^{i,j}_r=\{u_{\lambda\lfloor r(n+2)\rfloor}\in V(P'_{i,j})\mid 
\lambda\in\mathbb{N}_0\}$, i.e., these sets contain vertices of consecutive 
distance $r$ on the respective paths, after rounding down. Now let 
$H_r=\bigcup_{i,j\in[\kappa]} V(G_{i,j})\cup 
\bigcup_{i\in[\kappa],j\in[\kappa-1]} 
U^{i,j}_r\cup\bigcup_{i\in[\kappa-1],j\in[\kappa]} \tilde U^{i,j}_r$, so that 
every path of length more than~$r$ contains a hub of $H_r$. Any ball 
$B_{cr}(v)$ 
of radius $cr<c$ intersects only $O(c^2)$ gadgets~$G_{i,j}$, as observed above. 
As the edges of a cycle $O_{i,j}$ have length~$1$, the ball $B$ contains only 
$O(c)$ vertices of $O_{i,j}$. Thus $B_{cr}(v)$ contains $O(c)$ hubs of 
$V(G_{i,j})\cup U^{i,j}_r\cup \tilde U^{i,j}_r$ for each of the $O(c^2)$ 
gadgets 
$G_{i,j}$ it intersects. For constant $c$, this proves the claim.
\end{proof}

\begin{lem}\label{lem:tw}
The graph $G_\mc{I}$ has pathwidth at most~$\kappa+O(1)$.
\end{lem}
\begin{proof}
We construct a path decomposition of $G_\mc{I}$ using bags of size 
$\kappa+O(1)$. For each $i,j\in[\kappa]$ we define the sets 
$X^2_{i,j}=\{x^2_{i',j}\mid i'\in[i]\}$ and $X^4_{i,j}=\{x^4_{i',j}\mid 
i'\in[\kappa]\setminus[i-1]\}$ and let 
$K_{i,j}=\{y_{i,j},x^1_{i,j},x^3_{i,j}\}\cup X^2_{i,j}\cup X^4_{i,j}$ be a bag. 
Intuitively, these bags decompose the graph $G_\mc{I}$ ``from left to right'' 
according to \cref{fig:red}. More precisely, using some additional 
intermediate bags, the constructed path decomposition will arrange these bags 
on 
a path with start vertex $K_{1,1}$, such that traversing the path will 
consecutively move from $K_{i,j}$ to $K_{i+1,j}$ for each $1\leq i\leq\kappa-1$ 
and $1\leq j\leq\kappa$, and from $K_{\kappa,j}$ to $K_{1,j+1}$ for each $1\leq 
j\leq\kappa-1$.

To define the intermediate bags, consider a bag $K_{i,j}$ and note that the 
three connected components left after removing all vertices of $K_{i,j}$ from 
$G_\mc{I}$ are (a) the cycle $O_{i,j}$, (b) the subgraph~$L_{i,j}$ ``to the 
left 
of'' $K_{i,j}$ induced by all gadgets $G_{i',j'}$ and paths 
$P_{i',j'},P'_{i',j'}$ for which $j'\leq j-1$ and $i'\leq\kappa$, but also the 
gadgets $G_{i',j'}$ and paths $P'_{i',j'}$  for which $j'=j$ and $i'\leq i-1$, 
and finally (c) the subgraph $R_{i,j}$ ``to the right of'' $K_{i,j}$ induced by 
all gadgets $G_{i',j'}$ and paths $P_{i',j'},P'_{i',j'}$ for which either 
$j'=j$ 
and $i'\geq i+1$, or $j'\geq j+1$ and $i'\leq\kappa$, but also the paths 
$P_{i',j}$ where $i'\leq i$ and the path~$P'_{i,j}$. 

For any $i\leq\kappa-1$, removing the union $K_{i,j}\cup K_{i+1,j}$ from 
$G_\mc{I}$ leaves $L_{i,j}$, $R_{i+1,j}$, $O_{i,j}$, $O_{i+1,j}$, and the path 
$P'_{i,j}$ connecting the gadgets $G_{i,j}$ and $G_{i+1,j}$. The intermediate 
bags connecting $K_{i,j}$ and $K_{i,j+1}$ for $i\leq\kappa-1$ on the path 
decomposition will first cover $O_{i,j}$ and then $P'_{i,j}$: if 
$O_{i,j}=(v_1,v_2,\ldots,v_{16n^2+4},v_1)$, we define a sequence of bags 
$K^\tau_{i,j}=K_{i,j}\cup\{v_1,v_\tau,v_{\tau+1}\}$ where $\tau\in[16n^2+3]$, 
and if $P'_{i,j}=(u_0,u_1,\ldots,u_{n+2})$ where $u_0=x^3_{i,j}$ and 
$u_{n+2}=x^1_{i+1,j}$ then we define a sequence of bags 
$J^\tau_{i,j}=K_{i,j}\cup\{u_{\tau-1},u_{\tau}\}$ for $\tau\in[n+2]$. Note that 
for every edge~$e$ of $O_{i,j}$ or $P'_{i,j}$ there is a bag containing the 
vertices of $e$. Moreover, for every other edge $e$ of gadget $G_{i,j}$ 
connecting $O_{i,j}$ to $x^q_{i,j}$ for $q\in[4]$ or to $y_{i,j}$ there also is 
a bag~$K^\tau_{i,j}$ containing the vertices of $e$. Now, the path 
decomposition 
contains a subpath between the vertices corresponding to $K_{i,j}$ 
and~$K_{i+1,j}$, which starting from $K_{i,j}$ first traverses vertices for 
$K^\tau_{i,j}$ with increasing index~$\tau$, then connects $K^{16n^2+3}_{i,j}$ 
to~$J^1_{i,j}$, then traverses through $J^\tau_{i,j}$ with increasing $\tau$, 
and finally connects $J^{n+2}_{i,j}$ to~$K_{i+1,j}$. That is, the sequence of 
bags defined by the subpath is
\[
 (K_{i,j}, K^1_{i,j}, K^2_{i,j}, \ldots, K^{16n^2+3}_{i,j}, J^1_{i,j}, 
J^2_{i,j}, \ldots, J^{n+2}_{i,j}, K_{i+1,j}).
\]

To connect $K_{\kappa,j}$ to $K_{1,j+1}$ for some $j\leq\kappa-1$, we define 
additional bags $K'_{i,j}=X^2_{i,j}\cup X^4_{i,j+1}$. Starting from 
$K_{\kappa,j}$ and using intermediate bags, the path decomposition will 
traverse 
the bags $K'_{i,j}$ with decreasing index $i$ until reaching~$K_{1,j+1}$. 

We first describe the bags of the path decomposition connecting $K_{\kappa,j}$ 
to the first additional bag $K'_{\kappa,j}$. Defining the intermediate bags is 
similar to above. For any~$i\in[\kappa]$, removing the vertices of $K'_{i,j}$ 
from $G_\mc{I}$ leaves three connected components of which one is~$P_{i,j}$ 
connecting the respective gadgets $G_{i,j}$ and $G_{i,j+1}$, one is a component 
$L'_{i,j}$, which is $L_{1,j+1}$ without the paths $P_{i',j}$ where $i'\leq i$, 
and one is a component~$R'_{i,j}$, which is $R_{\kappa,j}$ without the paths 
$P_{i',j}$ where $i'\geq i$. If $P_{i,j}=(u_0,u_1,\ldots,u_{n+2})$ where 
$u_0=x^2_{i,j}$ and $u_{n+2}=x^4_{i,j+1}$, we define a sequence of bags 
$I^\tau_{i,j}=K'_{i,j}\cup \{u_{\tau-1},u_\tau\}$ for $\tau\in[n+2]$. Note that 
for every edge of $P_{i,j}$ there is a bag $I^\tau_{i,j}$ containing its 
vertices. Now, the path decomposition contains a subpath connecting vertices 
corresponding to $K_{\kappa,j}$ and $K'_{\kappa,1}$, which starting from 
$K_{\kappa,j}$ moves to $K^1_{\kappa,j}$, then through $K^\tau_{\kappa,j}$ with 
increasing index~$\tau$ to cover $O_{\kappa,j}$, and then connects 
$K^{16n^2+3}_{i,j}$ to $I^1_{\kappa,j}$. It then traverses through 
$I^\tau_{\kappa,j}$ with increasing index~$\tau$ to cover $P_{\kappa,j}$, after 
which it moves on to $K'_{\kappa,j}$. That is, the sequence of bags defined by 
the subpath is
\[
(K_{\kappa,j}, K^1_{\kappa,j}, K^2_{\kappa,j}, \ldots, K^{16n^2+3}_{\kappa,j}, 
I^1_{\kappa,j}, I^2_{\kappa,j}, \ldots, I^{n+2}_{\kappa,j}, \ldots, 
K'_{\kappa,j}).
\]

For any $i\leq \kappa-1$ the path decomposition contains a subpath connecting 
$K'_{i+1,j}$ to~$K'_{i,j}$, covering $P_{i,j}$ via the bags $I^\tau_{i,j}$ with 
increasing index $\tau$. The sequence defined by this subpath is
\[
(K'_{i+1,j}, I^1_{i,j}, I^2_{i,j}, \ldots, I^{n+2}_{i,j}, K'_{i,j}).
\]
The last additional bag $K'_{1,j}$ is connected directly to $K_{1,j+1}$ on the 
path decomposition.

Finally, when at $K_{\kappa,\kappa}$ the path decomposition only needs to cover 
$O_{\kappa,\kappa}$ to finish, i.e., to make sure that every vertex of 
$G_\mc{I}$ is contained in some bag. This can be done using the sequence 
$K^\tau_{\kappa,\kappa}$ with increasing index $\tau$, as above. That is, the 
sequence of bags defined by the final subpath of the path decomposition is
\[
(K_{\kappa,\kappa}, K^1_{\kappa,\kappa}, K^2_{\kappa,\kappa}, \ldots, 
K^{16n^2+3}_{\kappa,\kappa})
\]

As argued above, for every edge of $G_\mc{I}$ there is a bag containing its 
vertices. To argue that all bags containing some vertex of $G_\mc{I}$ form a 
subpath of the path decomposition, note that for intermediate bags 
$K^\tau_{i,j}$, $J^\tau_{i,j}$, and $I^\tau_{i,j}$ we have $K^\tau_{i,j}, 
J^\tau_{i,j}\supseteq K_{i,j}$ and $I^\tau_{i,j}\supseteq K'_{i,j}$. Also note 
that every vertex of $X=\{y_{i,j},x^q_{i,j}\mid q\in[4] \text{ and } 
i,j\in[\kappa]\}$ lies in some bag $K_{i,j}$ or $K'_{i,j}$. Let $x\in X$ be a 
vertex that appears in a bag $B$ but not in bag $B'$, and assume first that $B$ 
comes before~$B'$ in the sequence defined by the path decomposition. Since the 
path decomposition traverses $G_\mc{I}$ ``from left to right'', this means 
that $x$ lies in the set $L_{i,j}$ of some bag~$K_{i,j}$, or the set~$L'_{i,j}$ 
of some bag~$K'_{i,j}$, for which $B'\supseteq K_{i,j}$ or $B'\supseteq 
K'_{i,j}$, respectively. Similarly, if $B'$ comes before $B$ in the sequence, 
then $x$ lies in the set $R_{i,j}$ of some bag $K_{i,j}$, or the set $R'_{i,j}$ 
of some bag $K'_{i,j}$, for which $B'\supseteq K_{i,j}$ or $B'\supseteq 
K'_{i,j}$, respectively. Now observe that $L_{i,j}\subset L_{i',j'}$ and 
$R_{i,j}\supset R_{i',j'}$ for any $i<i'$ where $j=j'$ but also for any~$j<j'$, 
while for bags~$K'_{i,j}$, for any $j$ and $i>i'$ we have $L_{\kappa,j}\subset 
L'_{i,j}\subset L'_{i',j}$ and $R'_{i,j}\supset R'_{i',j}\supset R_{1,j}$. This 
means that, by definition of the sequence of bags along the path decomposition, 
if $B$ comes before $B'$ the vertex $x$ cannot appear in any bag after $B'$ 
either, while if $B'$ comes before $B$ then $x$ cannot appear in any bag before 
$B'$ either. As a consequence, the bags containing any $x\in X$ must form a 
subpath of the path decomposition.

It remains to argue about vertices not in $X$. Note that these only occur in 
intermediate bags $K^\tau_{i,j}$, $J^\tau_{i,j}$, and $I^\tau_{i,j}$ on some 
cycle of a gadget or a path connecting gadgets. Furthermore, the vertices of a 
cycle $O_{i,j}$ only occur in the bags~$K^\tau_{i,j}$, vertices of paths 
$P'_{i,j}$ (except the endpoints which lie in $X$) only occur in the 
bags~$J^\tau_{i,j}$, and vertices of paths $P_{i,j}$ (except the endpoints 
which 
lie in $X$) only occur in the bags $I^\tau_{i,j}$. It is thus easy to see from 
the definition of the sequences of bags above, that any vertex not in $X$ only 
lies in bags that form a subpath of the path decomposition.

Note that each bag contains $\kappa+O(1)$ vertices, which concludes the proof.
\end{proof}

The reduction given in \cref{sec:reduction} together with \cref{lem:dd}, 
\cref{lem:hd}, and \cref{lem:tw} imply \cref{thm:hardness}, since the 
\gti problem is W[1]-hard~\cite{pc-book} for parameter~$\kappa$, and we may 
assume w.l.o.g.\ that~$n\geq 2$. Moreover $\kappa=\Theta(\sqrt{k})$ and, under 
ETH, \gti has no $f(\kappa)\cdot n^{o(\kappa)}$ time algorithm~\cite{pc-book} 
for any computable function~$f$.

\section{An algorithm for low doubling metrics}
\label{sec:alg}

In this section we give a simple algorithm that generalizes one of 
\citet{agarwal2002clustering}, which for $D$-dimensional $L_q$ metrics 
computes a $(1+\eps)$-approximation in time~$f(\eps,k,D)\polyn$. In particular, 
any such metric has doubling dimension $O(D)$. Here we assume that the input 
metric has doubling dimension~$d$. A fundamental observation about metrics of 
bounded doubling dimension is the following, which can be proved by a simple 
recursive application of \cref{dfn:dd}. Here the aspect ratio of a set 
$Y\subseteq X$ is the diameter of $Y$ divided by the minimum distance between 
any two points of $Y$.

\begin{lem}[\cite{gupta2003doubling}]\label{lem:aspect}
Let $(X,\dist)$ be a metric with doubling dimension $d$ and $Y\subseteq X$ be a 
subset with aspect ratio~$\alpha$. Then $|Y|\leq 2^{d\lceil\log_2\alpha\rceil}$.
\end{lem}

To compute a $(1+\eps)$-approximation to \kc given a graph $G$ with vertex set 
$V$, we first compute its shortest-path metric $(V,\dist)$. We then compute 
several \emph{nets} of this metric, which are defined as follows.

\begin{dfn}
\label{dfn:netcover}
For a metric $(X,\dist)$, a subset $Y\subseteq X$ is called a 
\emph{$\delta$-cover} if for every $u\in X$ there is a $v\in Y$ such that 
$\dist(u,v)\leq \delta$. A \emph{$\delta$-net} is a $\delta$-cover with the 
additional property that $\dist(u,v)>\delta$ for all distinct points~$u,v\in 
Y$. 
\end{dfn}

Note that a $\delta$-net can be computed greedily in polynomial time. The first 
step of our algorithm is to guess the optimum cost $\rho$ by trying each of the 
$n\choose 2$ possible values. For each guess we compute an 
$\frac{\eps\rho}{2}$-net~$Y\subseteq V$. We know that the metric~$(V,\dist)$ 
can 
be covered by $k$ balls of diameter $2\rho$ each, which means that the aspect 
ratio of $Y$ inside of each ball is at most $4/\eps$. Thus by 
\cref{lem:aspect}, each ball contains $1/\eps^{O(d)}$ vertices of $Y$, and 
so 
$|Y|\leq k/\eps^{O(d)}$. 

An optimum \kc solution $C\subseteq Y$ for $(Y,\dist)$ can be computed by brute 
force in ${|Y| \choose k}=k^k/\eps^{O(kd)}$ steps. Since every center of the 
optimum solution $C^*\subseteq V$ of the input graph has a net point of $Y$ at 
distance at most~$\frac{\eps\rho}{2}$, there exists a \kc solution in $Y$ of 
cost at most $(1+\eps/2)\rho$, given that $\rho$ is the optimum cost. The 
computed center set $C\subseteq Y$ thus also has cost at most $(1+\eps/2)\rho$. 
Therefore $C$ covers all of $V$ with balls of radius $(1+\eps)\rho$, since 
every 
vertex of $V$ is at distance $\frac{\eps\rho}{2}$ from some vertex of~$Y$. Thus 
$C$ is a $(1+\eps)$-approximation of the input graph. Considering the guessed 
values of $\rho$ in increasing order, outputting the first computed solution 
with cost at most $(1+\eps)\rho$ gives the algorithm of \cref{thm:alg}.

\paragraph*{Acknowledgements.} We would like to thank the anonymous reviewers, 
who greatly helped to improve the quality of this manuscript.

\bibliography{papers}

\end{document}